\documentclass{article}

\pdfoutput=1

\usepackage{times}
\usepackage{graphicx} 
\usepackage{subfigure}

\usepackage{natbib}

\usepackage{algorithm}
\usepackage{algorithmic}

\usepackage{hyperref}

\usepackage[accepted]{icml2014}

\usepackage{amsmath}
\usepackage{amsfonts}
\usepackage{amsthm}

\newtheorem{proposition}{Proposition}

\DeclareMathOperator{\cov}{cov}

\icmltitlerunning{Exact Inference for Gaussian Process Regression in case of Big Data with the Cartesian Product Structure}

\begin{document}

\twocolumn[
\icmltitle{Exact Inference for Gaussian Process Regression in case of Big Data with the Cartesian Product Structure}

\icmlauthor{Belyaev Mikhail}{mikhail.belyaev@datadvance.net}
\icmladdress{Institute for Information Transmission Problems,
             Bolshoy Karetny per. 19, Moscow, 127994, Russia}
\icmladdress{DATADVANCE, llc, Pokrovsky blvd. 3, Moscow, 109028, Russia}
\icmladdress{PreMoLab, MIPT, Institutsky per. 9, Dolgoprudny, 141700, Russia}
\icmlauthor{Burnaev Evgeny}{evgeny.burnaev@datadvance.net}
\icmladdress{Institute for Information Transmission Problems,
             Bolshoy Karetny per. 19, Moscow, 127994, Russia}
\icmladdress{DATADVANCE, llc, Pokrovsky blvd. 3, Moscow, 109028, Russia}
\icmladdress{PreMoLab, MIPT, Institutsky per. 9, Dolgoprudny, 141700, Russia}
\icmlauthor{Kapushev Yermek}{ermek.kapushev@datadvance.net}
\icmladdress{Institute for Information Transmission Problems,
             Bolshoy Karetny per. 19, Moscow, 127994, Russia}
\icmladdress{DATADVANCE, llc, Pokrovsky blvd. 3, Moscow, 109028, Russia}

\icmlkeywords{Gaussian process, large dataset, ICML}

\vskip 0.3in
]

\begin{abstract}
Approximation algorithms are widely used in many engineering problems.
To obtain a data set for approximation a factorial design of experiments
is often used.
In such case the size of the data set can be very large.
Therefore, one of the most popular algorithms for approximation ---
Gaussian Process regression --- can be hardly applied due to its computational complexity.
In this paper a new approach for Gaussian Process regression in case of factorial design
of experiments is proposed.
It allows to efficiently compute exact inference and handle large multidimensional data sets.
The proposed algorithm provides fast and accurate approximation and also handles
anisotropic data.
\end{abstract}

\section{Introduction}
\label{submission}

Gaussian Processes (GP) have become a popular tool for regression which has lots of applications
in engineering problems \cite{rasmussen2006gaussian}.
They combine flexibility of being able to approximate a wide range of smooth functions
with simple structure of Bayesian inference and interpretable hyperparameters.

GP regression algorithm has $\mathcal{O}(N^3)$ time complexity and $\mathcal{O}(N^2)$
memory complexity, where $N$ is a size of the training sample.
For large training sets (ten thousands or more) construction of GP regression becomes intractable
problem on current hardware.

There is significant amount of research concerning sparse approximation of GP regression
reducing run-time complexity to $\mathcal{O}(M^2N)$ for some $M \ll N$ (for example, $M$ can be
the size of a subsample used for sparse approximation).
There are also methods based on Mixture of GPs and Bayesian Machine Committee.
Overview of these methods can be found in
\cite{rasmussen2006gaussian, rasmussen2001infinite, quinonero2005unifying}.

Reduced run-time and memory complexity can be achieved not only by means of
sparse approximations
and Mixtures of GPs
but also by taking into account a structure of a design of experiments.
In engineering problems they often use a factorial design \cite{Montgomery2006DAE}.
That is there are several groups of variables, in each group variables take values from some discrete finite set.
Such sets and corresponding groups of variables are called {\em factors}.
Number of different values in a factor is called {\em factor size} and the values themselves are called {\em levels}.
The Cartesian product of factors forms the training set.
When the factorial design of experiments is used the size of the data set can be very large (as it grows exponentially with dimension of input variables).

There are several methods based on splines which consider this special structure of the given data \cite{stone97polynomialsplines}.
A disadvantage of these methods is that they work only with one-dimensional factors
and can't be applied to a more general case when factors are multidimensional.
Another shortcoming is that such approaches don't have approximation accuracy evaluation procedure.

There are also several approaches for GP regression on a lattice based on block Toeplitz covariance matrix with
Toeplitz blocks and circulant embedding \cite{zimmerman1989, woodChan1994, Dietrich1997}.
Such methods have $\mathcal{O}(N \log N)$ time complexity and $\mathcal{O}(N)$ memory complexity.
They can be used only if all the factors are one-dimensional and after monotonic transformation of each factor
(levels of each factor should be equally spaced).
However, in the case of multidimensional factors the covariance matrix doesn't have the desired structure (it should be block Toeplitz with Toeplitz blocks)
and therefore these approaches can't be generalized for such design of experiments.

There is another problem which we are likely to encounter.
Factor sizes can vary significantly.
Engineers usually use large factors sizes if corresponding input variables have big impact on function values
otherwise the factors sizes are likely to be small,
i.e. the factor sizes are often selected using knowledge from a subject domain \cite{rendall2008aircraftSurface}.
For example, if it is known that dependency on some variable is quadratic
then the size of this factor will be 3 as a larger size is redundant.
Difference between factor sizes can lead to degeneracy of the GP model.
We will refer to this property of data set as {\em anisotropy}.

In this paper we describe an approach that takes into
account factorial nature of the design of experiments in general case of multidimensional factors
and allows to efficiently calculate exact inference of GP regression.
We also discuss how to choose initial values of parameters for the GP model and regularization in order to
take into account possible anisotropy of the training data set.

\subsection{Approximation problem}
Let $f(\mathbf{x})$ be some unknown smooth function.
The task is given a data set $\mathcal{D} = \{(\mathbf{x}_i, y_i), \mathbf{x}_i \in \mathbb{R}^d, y_i \in \mathbb{R}\}_{i = 1}^N$
of $N$ pairs of inputs $\mathbf{x}_i$ and outputs $y_i$ construct an approximation $\hat{f}(\mathbf{x})$
of the function $f(\mathbf{x})$
where outputs $y_i$ are assumed to be noisy with additive i.i.d. Gaussian noise:
\begin{equation}
    \label{eq:model}
      y_i = f(\mathbf{x}_i) + \varepsilon_i, \quad \varepsilon_i \sim \mathcal{N}(0, \sigma_{noise}^2).
\end{equation}

\subsection{Factorial design of experiments}
In this paper a special case is considered when the design of experiments is factorial.
Let us refer to sets of points
${s_k = \{ x_{i_k}^k \in {\rm X}_k \}_{i_k = 1}^{n_k}}$,
${\rm X}_k \subset \mathbb{R}^{d_k}, \, k = \overline{1, K}$ as {\em factors}.
A set of points $\mathbf{S}$ is referred to as a factorial design of experiments
if it is a Cartesian product of factors
\begin{equation}
  \label{eq:factorial_design}
  \begin{split}
    \mathbf{S} = & \, s_1 \times s_2 \times \cdots \times s_k = \\
    & \{ [x_{i_1}^1, \ldots, x_{i_K}^K], \{ i_k = 1, \ldots, n_k\}_{k = 1}^K \}.
  \end{split}
\end{equation}
The elements of $\mathbf{S}$ are vectors of a dimension $d = \sum_{i = 1}^K d_i$ and
the sample size is a product of sizes of all factors $N = \prod_{i = 1}^K n_i$.
If all the factors are one-dimensional $\mathbf{S}$ is a full factorial design.
But in a more general case factors are multidimensional (see example in Figure \ref{fig:multidim_factor}).
Note that in this paper we do not consider categorical variables,
factorial design is implemented across continuous real-valued features.

\begin{figure}
  \includegraphics[width=0.5\textwidth]{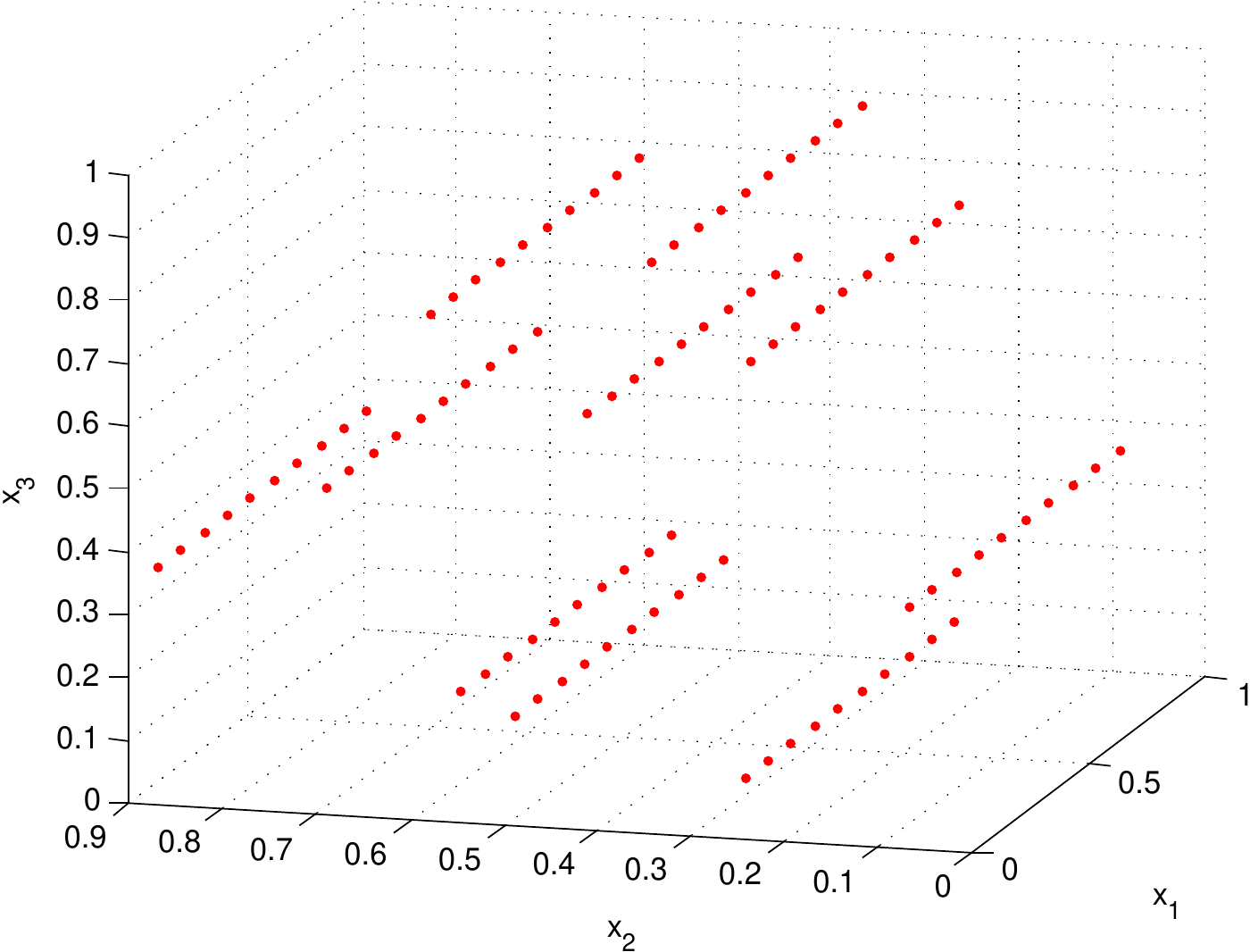}
  \caption{Example of a multidimensional factor. In the figure $x_1$ is a usual one-dimensional factor
  and $(x_2, x_3)$ is a 2-dimensional factor.}
  \label{fig:multidim_factor}
\end{figure}

\subsection{Gaussian Process Regression}
GP regression is a Bayesian approach where a prior distribution over continuous functions
is assumed to be a Gaussian Process, i.e.
\begin{equation}
  \label{eq:gp_model}
  \mathbf{f} \, | \, \mathbf{X} \sim \mathcal{N}(\boldsymbol{\mu}, \, \mathbf{K}_f),
\end{equation}
where $\mathbf{f} = (f(\mathbf{x}_1), f(\mathbf{x}_2), \ldots, f(\mathbf{x}_N))$ is a vector of outputs,
$\mathbf{X} = (\mathbf{x}_1^T, \mathbf{x}_2^T, \ldots, \mathbf{x}_N^T)^T$ is a matrix of inputs,
$\boldsymbol{\mu} = (\mu(\mathbf{x}_1), \mu(\mathbf{x}_2), \ldots, \mu(\mathbf{x}_N))$ is a mean vector for some function $\mu(\mathbf{x})$,
$\mathbf{K}_f = \{ k(\mathbf{x}_i, \mathbf{x}_j) \}_{i, j = 1}^N$ is a covariance matrix for some a priori selected covariance function $k$.


Without loss of generality we make the standard assumption of zero-mean data.
Also assume that we are given observations with Gaussian noise
\[
y_i = f(\mathbf{x}_i) + \varepsilon_i, \quad \varepsilon_i \sim \mathcal{N}(0, \sigma_{noise}^2)
\]
For a prediction of $f(\mathbf{x}_*)$ at new unseen data point $\mathbf{x}_*$ the posterior mean
conditioned on the observations $\mathbf{y} = (y_1, y_2, \ldots, y_N)$ at training inputs $\mathbf{X}$
is used
\begin{equation}
  \label{eq:posterior_mean}
  \hat{f}(\mathbf{x}_*) = \mathbf{k}(\mathbf{x}_*)^T \mathbf{K}_y^{-1}\mathbf{y},
\end{equation}
where $\mathbf{k}(\mathbf{x}_*) = (k(\mathbf{x}_*, \mathbf{x}_1), \ldots, k(\mathbf{x}_*, \mathbf{x}_N))^T$,
$\mathbf{K}_y = \mathbf{K}_f + \sigma_{noise}^2\mathbf{I}$ and
$\mathbf{I}$ is an identity matrix.
For approximation accuracy evaluation the posterior variance is used
\begin{equation}
  \label{eq:posterior_covariance}
  \cov(\hat{f}(\mathbf{x}_*)) = k(\mathbf{x}_*, \mathbf{x}_*) - \mathbf{k}(\mathbf{x}_*)^T \mathbf{K}_y^{-1} \mathbf{k}(\mathbf{x}_*).
\end{equation}
Usually for GP regression a squared exponential covariance function is used
\begin{equation}
  \label{eq:covariance_function}
  k(\mathbf{x}_p, \mathbf{x}_q) = \sigma_f^2 \exp \left ( -\sum_{i = 1}^d \theta_i^2 (\mathbf{x}_p^{(i)} - \mathbf{x}_q^{(i)})^2 \right ).
\end{equation}

Let us denote the vector of hyperparameters $\theta_i$ by $\boldsymbol{\theta} = (\theta_1, \ldots, \theta_d)$.
To choose the hyperparameters of our model we consider the log likelihood
\begin{equation}
  \label{eq:loglikelihood}
  \begin{split}
    \log p(\mathbf{y} \, | \, \mathbf{X}, \boldsymbol{\theta}, \sigma_f, \sigma_{noise}) = & -\frac12 \mathbf{y}^T \mathbf{K}_y^{-1}\mathbf{y} - \frac12 \log |\mathbf{K}_y| - \\
    & - \frac{N}{2} \log 2 \pi
  \end{split}
\end{equation}
and optimize it over the hyperparameters \cite{rasmussen2006gaussian}.
The runtime complexity of learning GP regression is $\mathcal{O}(N^3)$ as we need to calculate
inverse of $\mathbf{K}_y$, its determinant and derivatives of the log likelihood.

\section{Proposed approach}

\subsection{Tensor and related operations}
For further discussions we will use tensor notation, so let's introduce
definition of a tensor and some related operations.

A {\em tensor} $\mathcal{Y}$ is a $K$-dimensional matrix of size $n_1 \times n_2 \times \cdots \times n_K$
\cite{kolda09tensordecompositions}:
\begin{equation}
  \label{eq:tensor}
  \mathcal{Y} = \{ y_{i_1, i_2, \ldots, i_K}, \{i_k = 1, \ldots, n_k \}_{k = 1}^K \}.
\end{equation}

By $\mathcal{Y}^{(j)}$ we will denote a matrix consisting of elements of the tensor $\mathcal{Y}$
whose rows are $1 \times n_j$ slices of $\mathcal{Y}$ with fixed indices
$i_{j + 1}, \ldots, i_K, i_1, \ldots, i_{j - 1}$ and altering index $i_j = 1, \ldots, n_j$.
In case of $2$-dimensional tensor it holds that
${\mathcal{Y}^{(1)} = \mathcal{Y}^T}$ and $\mathcal{Y}^{(2)} = \mathcal{Y}$.

Now let's introduce multiplication of a tensor by a matrix along the direction $i$.
Let $B$ be some matrix of size $n_i \times n_i'$.
Then the product of the tensor $\mathcal{Y}$ and the matrix $B$ along the direction $i$
is a tensor $\mathcal{Z}$ of size $n_1 \times \cdots \times n_{i - 1} \times n_i' \times n_{i + 1} \times \cdots \times n_K$
such that $\mathcal{Z}^{(i)} = \mathcal{Y}^{(i)}B$.
We will denote this operation by $\mathcal{Z} = \mathcal{Y} \otimes_i B$.
For a $2$-dimensional tensor $\mathcal{Y}$ multiplication along the first direction is a left multiplication
by matrix $\mathcal{Y} \otimes_1 B = B^T \mathcal{Y}$,
and along the second direction --- is a right multiplication $\mathcal{Y} \otimes_2 B = \mathcal{Y} B$.

Let's consider an operation {\em vec} which for every multidimensional matrix $\mathcal{Y}$
returns a vector containing all elements of $\mathcal{Y}$.
An inner product of tensors $\mathcal{Y}$ and $\mathcal{Z}$ is the inner product of vectors
${\rm vec}(\mathcal{Y})$ and ${\rm vec}(\mathcal{Z})$
\[
\left < \mathcal{Y}, \mathcal{Z}\right > = \left < {\rm vec}(\mathcal{Y}), {\rm vec}(\mathcal{Z}) \right >.
\]

\subsection{The Kronecker product and its properties}
Some properties of the Kronecker product will be used to compute inference, so let us introduce them in this section.

Let $A$ and $B$ be $m \times n$ and $p \times q$ matrices.
The {\em Kronecker product} of matrices $A$ and $B$, denoted by $\otimes$, is a block matrix
\[
A \otimes B = \begin{bmatrix} a_{11} B & \cdots & a_{1n}B \\ \vdots & \ddots & \vdots \\ a_{m1} B & \cdots & a_{mn} B \end{bmatrix},
\]
where $a_{ij}$ is a $(i, j)$-th element of matrix $A$.
Basic properties of the Kronecker product:
\begin{enumerate}
  \item $A \otimes (B + C) = A \otimes B + A \otimes C$.
  \item $A \otimes (B \otimes C) = (A \otimes B) \otimes C$.
  \item $(A \otimes B)^T = A^T \otimes B^T$.
  \item \label{prop:kronecker_inverse}  $(A \otimes B)^{-1} = A^{-1} \otimes B^{-1}$.
  \item $(A \otimes B)(C \otimes D) = AC \otimes BD$.
\end{enumerate}

The Kronecker product is closely related to the multiplication of a tensor by a matrix.
For every multidimensional matrix $\mathcal{Y}$ of size $n_1 \times n_2 \times \cdots \times n_K$
and $n_i \times p_i$ size matrices $B_i$, $i = 1, \ldots, K$ the following identity holds \cite{kolda09tensordecompositions}
\begin{equation}
  \label{eq:kronecker_product}
  \begin{split}
    (B_1 \otimes B_2 \cdots \otimes & B_K){\rm vec}(\mathcal{Y}) = \\
    & {\rm vec} (\mathcal{Y} \otimes_1 B_1^T \cdots \otimes_K B_K^T),
  \end{split}
\end{equation}
Let's compare a complexity of the right and the left hand sides of (\ref{eq:kronecker_product}).
For simplicity we assume that all the matrices $B_i$ are quadratic of size $n_i \times n_i$ and
$N = \prod n_i$.
Then computation of the left hand side of (\ref{eq:kronecker_product}) requires
$N^2$ operations (of additions and multiplications) not taking into account complexity of the Kronecker product
while the right hand side requires $N \sum_i n_i$ operations.

\subsection{Computing inference}

In this section an efficient way to compute inverse of a covariance matrix will be described
as well as calculation of the log likelihood, its derivatives, the predictive mean and the covariance matrix using introduced tensor operations.

Covariance function (\ref{eq:covariance_function}) can be represented as
a product of covariance functions each depending only on variables from one factor
\begin{equation}
  \label{eq:covariance_function_general}
  k(\mathbf{x}_p, \mathbf{x}_q) = \prod_{i = 1}^K k_i(x_p^i, x_q^i),
\end{equation}
where $x_p^i, x_q^i \in \mathbb{R}^{d_i}$ belong to the same factor $s_i$.
For the squared exponential function we have $k_i(x_p^i, x_q^i) = \sigma_{f, i}^2\exp \left ( -\sum_j^{d_i} \left (\theta_i^{(j)} \right )^2
  \big (x_p^{(j), i} - x_q^{(j), i} \big )^2 \right )$,
where $x_p^{(j), i}$ is a $j$-th component of $x_p^i$.
Note that in general case covariance functions $k_i$ are not necessarily squared exponential, they can be of different types for different factors.
It allows to take into account special features of factors (knowledge from a subject domain) if they are known.
In such case the function defined by (\ref{eq:covariance_function_general}) is still a valid
covariance function being the product of separate covariance functions.
From now on we will denote by $\theta_i = (\theta_i^{(1)}, \ldots, \theta_i^{(d_i)})$ the set of hyperparameters for covariance
function of the $i$-th factor and let $\boldsymbol{\theta} = (\theta_1, \ldots, \theta_K)$.

Such form of the covariance function and the factorial design of experiments
allows us to represent the covariance matrix as the Kronecker product
\begin{equation}
  \label{eq:covariance_kronecker}
  \mathbf{K}_f = \bigotimes_{i = 1}^K\mathbf{K}_i,
\end{equation}
where $\mathbf{K}_i$ is a covariance matrix defined by the $k_i$ covariance function
computed at points from the $i$-th factor $s_i$.

The Kronecker product of matrices can be efficiently inverted due to property \ref{prop:kronecker_inverse}.
To invert the matrix $\mathbf{K}_y = \mathbf{K}_f + \sigma_{noise}^2 \mathbf{I}$
we will use the Singular Value Decomposition (SVD)
\[
\mathbf{K}_i = \mathbf{U}_i \mathbf{D}_i \mathbf{U}_i^T,
\]
where $\mathbf{U}_i$ is an orthogonal matrix of eigenvectors of matrix $\mathbf{K}_i$
and $\mathbf{D}_i$ is a diagonal matrix of eigenvalues.
Using the properties of the Kronecker product and representing an identity matrix as
$\mathbf{I}_{d_i} = \mathbf{U}_i \mathbf{U}_i^T$ we obtain
\begin{equation}
  \label{eq:inverse_covariance}
  \mathbf{K}_y^{-1} = \left (\bigotimes_{i = 1}^K \mathbf{U}_i \right ) \left (
    \left [\bigotimes_{i = 1}^K \mathbf{D}_i \right ] + \sigma_{noise}^2 \mathbf{I} \right )^{-1}
  \left ( \bigotimes_{i = 1}^K \mathbf{U}_i^T \right ).
\end{equation}

Computing SVD for all $\mathbf{K}_i$ requires $\mathcal{O}(\sum_k n_k^3)$ operations.
Calculation of the Kronecker product in (\ref{eq:inverse_covariance}) has complexity $\mathcal{O}(N^2)$.
So, this gives us overall complexity $\mathcal{O}(N^2)$ for calculation of expressions for the log likelihood,
the predictive mean and the covariance.
It is faster than the straightforward calculations, however it can be improved.

Equations (\ref{eq:posterior_mean}), (\ref{eq:posterior_covariance}), (\ref{eq:loglikelihood})
for GP regression do not require explicit inversion of $\mathbf{K}_y$.
In each equation it is multiplied by vector $\mathbf{y}$ (or $\mathbf{k}_*$).
So, we will compute $\mathbf{K}_y^{-1} \mathbf{y}$ instead of explicitly inverting
$\mathbf{K}_y$ and then multiplying it by the vector $\mathbf{y}$.

Let $\mathcal{Y}$ be a tensor containing values of the vector $\mathbf{y}$ such that
${\rm vec}(\mathcal{Y}) = \mathbf{y}$.
Now using identities (\ref{eq:kronecker_product}) and (\ref{eq:inverse_covariance})
we can write $\mathbf{K}_y^{-1}\mathbf{y}$ as
\begin{align}
  \label{eq:fast_inverse_covariance}
    \mathbf{K}_y^{-1}\mathbf{y} = \left ( \bigotimes_{i = 1}^K \mathbf{U}_i \right )
    \left ( \left [\bigotimes_{i = 1}^K \mathbf{D}_i \right ] + \sigma_{noise}^2 \mathbf{I} \right )^{-1} \times \nonumber \\
    \times {\rm vec}(\mathcal{Y} \otimes_1 \mathbf{U}_1 \cdots \otimes_K \mathbf{U}_K) = \nonumber \\
    = {\rm vec}\left [ \left ((\mathcal{Y} \otimes_1 \mathbf{U}_1 \cdots \otimes_K \mathbf{U}_K) * \mathcal{D}^{-1} \right ) \right .\otimes_1 \nonumber \\
    \left . \otimes_1 \mathbf{U}_1^T \cdots \otimes_K \mathbf{U}_K^T \right ],
\end{align}
where $\mathcal{D}$ is a tensor constructed by transforming the diagonal of matrix
$\big [\bigotimes_k \mathbf{D}_k \big ] + \sigma_{noise}^2 \mathbf{I}$ into a tensor.

The elements of the tensor $\mathcal{D}$ are eigenvalues of the matrix $\mathbf{K}_y$,
therefore its determinant can be calculated as
\begin{equation}
  \label{eq:fast_determinant}
  |\mathbf{K}_y| = \prod_{i_1, \ldots, i_K} \mathcal{D}_{i_1, \ldots, i_K}.
\end{equation}

\begin{proposition}
  The computational complexity of the log likelihood (\ref{eq:loglikelihood}), where $\mathbf{K}_y^{-1}y$ and
  $|\mathbf{K}_y|$ are calculated using (\ref{eq:fast_inverse_covariance}) and (\ref{eq:fast_determinant}), is
  \begin{equation}
    \mathcal{O} \left (\sum_{i = 1}^K n_i^3 + N\sum_{i = 1}^K n_i \right ).
  \end{equation}
\end{proposition}
\begin{proof}
  Let's calculate the complexity of computing $\mathbf{K}_y^{-1} \mathbf{y}$ using (\ref{eq:fast_inverse_covariance}).
  Computation of the matrices $\mathbf{U}_i$ and $\mathbf{D}_i$ requires $\mathcal{O}(\sum_i n_i^3)$ operations.
  Multiplication of the tensor $\mathcal{Y}$ by the matrices $\mathbf{U}_i$ requires $\mathcal{O}(N\sum_i n_i)$ operations.
  Further, component-wise product of the obtained tensor and the tensor $\mathcal{D}^{-1}$ requires $\mathcal{O}(N)$
  operations.
  And complexity of multiplication of the result by the matrices $\mathbf{U}_i$ is again $\mathcal{O}(N\sum_i n_i)$.
  The determinant, calculated by equation (\ref{eq:fast_determinant}), requires $\mathcal{O}(N)$ operations.
  Thus, the overall complexity of computing (\ref{eq:fast_inverse_covariance}) is
  ${\mathcal{O}(\sum_{i = 1}^K n_i^3 + N\sum_{i = 1}^K n_i)}$.
\end{proof}

For more illustrative estimate of the computational complexity
suppose that $n_i \ll N$ (number of factors is large and their sizes are close).
In this case it holds that
$\mathcal{O}(N\sum_i n_i) = \mathcal{O}(N^{1 + \frac{1}{K}})$ and this is
much less than $\mathcal{O}(N^3)$.

To optimize the log likelihood over the hyperparameters we use a gradient based method.
The derivatives of the log likelihood with respect to the hyperparameters take the form
\begin{equation}
  \label{eq:loglikelihood_derivative}
  \begin{split}
    \frac{\partial}{\partial \theta} \left (\log \vphantom{p(\mathbf{y} | \mathbf{X}, \sigma_f, \sigma_{noise})} \right . &
    \left . p(\mathbf{y} | \mathbf{X}, \sigma_f, \sigma_{noise}) \right ) = \\
    & -\frac12 {\rm Tr}(\mathbf{K}_y^{-1} \mathbf{K}') + \frac12\mathbf{y}^T \mathbf{K}_y^{-1} \mathbf{K}' \mathbf{K}_y^{-1} \mathbf{y},
  \end{split}
\end{equation}
where $\theta$ is one of the hyperparameters of covariance function (component of $\theta_i$, $\sigma_{noise}$ or $\sigma_{f, i}, i = 1, \ldots, d$) and
$\mathbf{K}' = \dfrac{\partial \mathbf{K}}{\partial \theta}$.
$\mathbf{K}'$ is also the Kronecker product
\[
\mathbf{K}' = \mathbf{K}_1 \otimes \cdots \otimes \mathbf{K}_{i - 1} \otimes \frac{\partial \mathbf{K}_i}{\partial \theta}
\otimes \mathbf{K}_{i + 1} \otimes \cdots \otimes \mathbf{K}_K,
\]
where $\theta$ is a parameter of the $i$-th covariance function.
Denoting by $\mathcal{A}$ a tensor such that ${\rm vec}(\mathcal{A}) = \mathbf{K}_y^{-1}\mathbf{y}$
the second term in equation (\ref{eq:loglikelihood_derivative}) can be efficiently computed
using the same technique as in (\ref{eq:fast_inverse_covariance}):
\begin{equation}
  \label{eq:fast_derivative_second_term}
  \begin{split}
    \frac12\mathbf{y}^T & \mathbf{K}_y^{-1} \mathbf{K}'\mathbf{K}_y^{-1} \mathbf{y} = \\
    & \left < \mathcal{A}, \mathcal{A} \otimes_1 \mathbf{K}_1^T \otimes_2 \cdots \otimes_{i - 1} \mathbf{K}_{i - 1}^T \otimes_i
      \vphantom{\frac{\partial \mathbf{K}_i^T}{\partial \theta}} \right .\\
    & \quad \left . \frac{\partial \mathbf{K}_i^T}{\partial \theta} \otimes_{i + 1} \mathbf{K}_{i + 1}^T \otimes_{i + 2} \cdots \otimes_K \mathbf{K}_K^T \right >.
  \end{split}
\end{equation}
The complexity of calculating this term of derivative is the same as the complexity of
equation (\ref{eq:fast_inverse_covariance}).

Now let's compute the first term
\begin{equation}
  \begin{split}
    {\rm Tr}(\mathbf{K}_y^{-1} \mathbf{K}') = &\; {\rm Tr} \left (
      \left (\bigotimes_{i = 1}^K \mathbf{U}_i \right ) \mathbf{D}^{-1} \left (\bigotimes_{i = 1}^K \mathbf{U}_i^T \right ) \mathbf{K}'
    \right)= \\
    = &\; {\rm Tr} \left (\mathbf{D}^{-1} \left (\bigotimes_{i = 1}^K \mathbf{U}_i^T \mathbf{K}_i' \mathbf{U}_i \right ) \right ) = \\
    = &\; \left < {\rm diag} \left (\mathbf{D}^{-1} \right ), {\rm diag} \left (\bigotimes_{i = 1}^K \mathbf{U}_i \mathbf{K}_i' \mathbf{U}_i \right ) \right > = \\
    = &\; \left < {\rm diag} \left (\mathbf{D}^{-1} \right ),  \bigotimes_{i = 1}^K {\rm diag} \left ( \mathbf{U}_i \mathbf{K}_i' \mathbf{U}_i \right ) \right >,
  \end{split}
\end{equation}
where ${\rm diag}\big(A \big)$ is a vector of diagonal elements of a matrix $A$,
$\; \mathbf{D} = \bigotimes_i \mathbf{D}_i + \sigma_{noise}^2 \mathbf{I}$.

The computational complexity of this derivative term is the same as the computational
complexity of equation (\ref{eq:fast_derivative_second_term}).

Thus, we obtain
\begin{proposition}
  The computational complexity of calculating derivatives of the log likelihood is
  ${\mathcal{O}\left (\sum\limits_{i = 1}^K n_i^3 + N \sum\limits_{i = 1}^K n_i \right )}$.
\end{proposition}

\begin{table}[t]
  \caption{Runtime (in seconds) of tensored GP and original GP algorithms.}
  \label{tb:runtime}
  \vskip 0.15in
  \begin{center}
    \begin{small}
      \begin{sc}
        \begin{tabular}{rrr}
          \hline
          \abovespace\belowspace
          & original GP   & tensored GP \\
          64     & 0.8    & 0.16 \\
          160    & 2.69   & 0.16 \\
          432    & 14.31  & 0.74 \\
          1000   & 120.38 & 1.02 \\
          2000   & 970.21 & 1.11 \\
          10240  & ---    & 33.18 \\
          64000  & ---    & 74.9 \\
          160000 & ---    & 175.15 \\
          400000 & ---    & 480.14\\
          \hline
          \abovespace\belowspace
        \end{tabular}
      \end{sc}
    \end{small}
  \end{center}
  \vskip -0.1in
\end{table}

Table \ref{tb:runtime} contains training times for
original GP and proposed GP regression for different sample sizes.
The experiments were conducted on a PC with Intel i7 2.8 GHz processor and
4 GB RAM.
For original GP we used GPML Matlab code \cite{gpmltoolbox}.
We also adopted GPML code to use tensor operations.
The results illustrate that the proposed approach is much more faster than the original GP
and allows to make approximations using extremely large data sets.

\subsection{Anisotropy}
In this section we will consider an anisotropy problem.
As it was mentioned in an engineering practice factorial designs are
often anisotropic, i.e. sizes of factors differ significantly.
It is a common case for the GP regression to become degenerate in such situation.
Suppose that the given design of experiments consists of two one-dimensional factors
with sizes $n_1$ and $n_2$.
Assume that $n_1 \ll n_2$.
Then one could expect the length-scale for the first factor to be much greater than
the length-scale for the second factor (or $\theta_1 \ll \theta_2$).
However, in practice we often observe the opposite $\theta_1 \gg \theta_2$.
This happens because the optimization algorithm stacks in a local maximum during maximization over
the hyperparameters as the objective function (the log likelihood) is non-convex with lots of local maxima.
We get an undesired effect of degeneracy:
in the region without training points the approximation is constant
and it has sharp peaks at training points.
This situation is illustrated in Figure \ref{fig:anisotropy_degeneracy}
(compare with the true function in Figure \ref{fig:anisotropy_degeneracy_true}).

Let us denote length-scales as $l_k^{(i)} = \big [\theta_k^{(i)} \big ]^{-1}$.
To incorporate our prior knowledge about factor sizes into regression model
we introduce prior distribution on the hyperparameters~$\boldsymbol{\theta}$:
\begin{equation}
  \label{eq:prior}
  \frac{\theta_k^{(i)} - a_k^{(i)}}{b_k^{(i)} - a_k^{(i)}} \sim \mathcal{B}e(\alpha, \beta), \, \{ i = 1, \ldots, d_k\}_{k = 1}^K,
\end{equation}
i.e. prior on hyperparameter $\theta_k^{(i)}$ is a beta distribution with parameters $\alpha$
and $\beta$ scaled to some interval $\left [a_k^{(i)}, b_k^{(i)} \right ]$.

The log likelihood then has the form
\begin{equation}
  \label{eq:loglikelihood_prior}
  \begin{split}
    \log p(\mathbf{y} \, | \, \mathbf{X}, \boldsymbol{\theta}, \sigma_f, & \sigma_{noise}) = -\frac12 \mathbf{y}^T \mathbf{K}_y^{-1}\mathbf{y} - \frac12 \log |\mathbf{K}_y| - \\
    - \frac{N}{2} \log 2 \pi + & \sum_{k, i} \left ( (\alpha - 1) \log \left (\frac{\theta_k^{(i)} - a_k^{(i)}}{b_k^{(i)} - a_k^{(i)}} \right ) + \right . \\
    + (\beta - 1) \log \left (1 -  \vphantom{\frac{\theta_k^{(i)} - a_k^{(i)}}{b_k^{(i)} - a_k^{(i)}}} \right . &  \left . \left .\frac{\theta_k^{(i)} - a_k^{(i)}}{b_k^{(i)} - a_k^{(i)}} \right )  \right ) - d \log({\rm B}(\alpha, \beta)),
\end{split}
\end{equation}
where ${\rm B}(\alpha, \beta)$ is a beta function.


Numerous references use gamma distribution as a prior, e.g. \cite{neal1997}.
Preliminary experiments showed that GP regression models with gamma prior often degenerates.
So, in this work we use prior with compact support.
By introducing such prior we restrict parameters $\theta_k^{(i)}$ to belong
to some interval $\left [a_k^{(i)}, b_k^{(i)} \right ]$ (or length-scales $l_k^{(i)}$ to belong to the interval
$\left [\big (b_k^{(i)} \big)^{-1}, \big (a_k^{(i)} \big )^{-1} \right ]$).
This choice of prior allows to prohibit too small and too large length-scales excluding possibility to degenerate
(if intervals of allowed length-scales are chosen properly).

It seems reasonable that for an approximation to fit the training points
the length-scale is not needed to be much less than the distance between points.
That's why we choose the lower bound for the length-scale $l_k^{(i)}$ to be $c_k * \min\limits_{x, y \in s_k, x^{(i)} \ne y^{(i)}} ||x^{(i)} - y^{(i)}||$
and the upper bound for the length-scale to be ${C_k * \max\limits_{x, y \in s_k} ||x^{(i)} - y^{(i)}||}$.
The value $c_k$ should be close to 1.
If it is chosen too small we are taking risks to overfit the data by allowing small length-scales.
If $c_k$ is too large we are going to underfit the data by allowing only large length-scales and forbidding small ones.
Constants $C_k$ must be much greater than $c_k$ to permit large length-scales and preserve flexibility.
In this work we used $c_k = 0.5$ and $C_k = 100$.
Such values of $c_k$ and $C_k$ worked rather good in our test cases.

Parameters of beta distribution was set to $\alpha = \beta = 2$ to get
symmetrical probability distribution function (see Figure \ref{fig:beta}).
We don't know a priori large or small should be the values of GP parameters,
so prior distribution should impose nearly the same penalties for the intermediate values of the parameters.
As it can be seen from Figure \ref{fig:beta} the chosen distribution does exactly what is needed.

Figure \ref{fig:nondegenerate_tensorGP} illustrates approximation
of the GP regression with introduced prior distribution (and initialization described in Section \ref{sec:initialization}).
The hyperparameters were chosen such that the approximation is non-degenerate.

\begin{figure}
  \centering
  \includegraphics[width=0.3\textwidth]{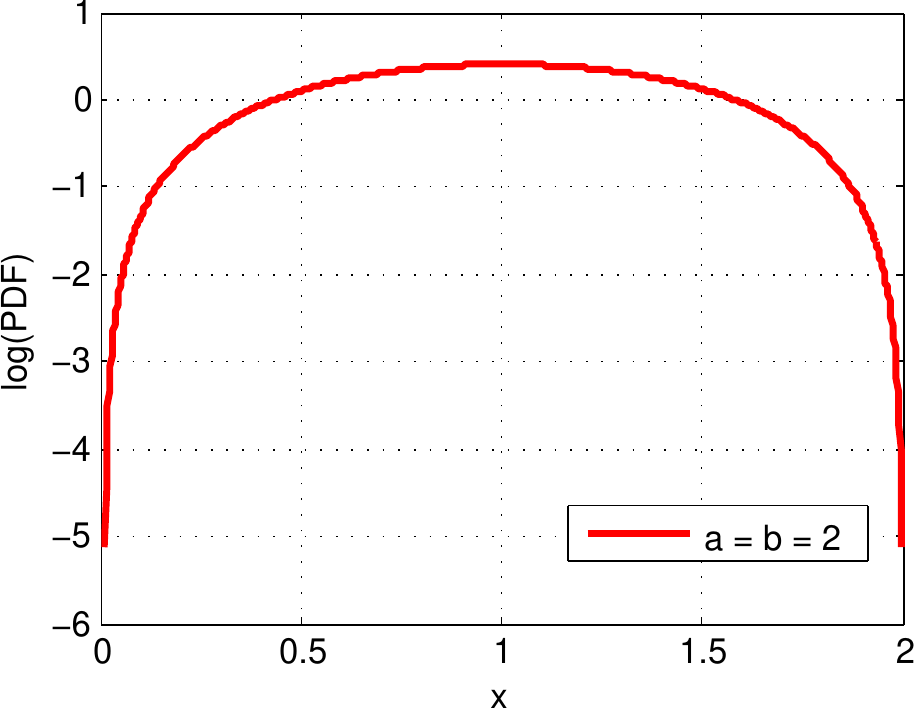}
  \caption{Logarithm of Beta distribution probability density function, rescaled to $[0.01, 2]$ interval,
    with parameters $\alpha = \beta = 2$.}
  \label{fig:beta}
\end{figure}

\subsection{Initialization}
\label{sec:initialization}
It is also important to choose reasonable initial values of hyperparameters
in order to converge to a good solution during parameters optimization.
The kernel-widths for different factors should have different scales because
corresponding factor sizes have different number of levels.
So it seems reasonable to use average distance between points in a factor as an initial value
\begin{equation}
  \label{eq:initialization}
  \theta_k^{(i)} = \left [ \frac{1}{n_k} \left ( \max\limits_{x \in s_k}(x^{(i)}) - \min\limits_{x \in s_k}(x^{(i)}) \right ) \right ]^{-1}.
\end{equation}

\begin{figure}
  \centering
  \includegraphics[width=0.4\textwidth]{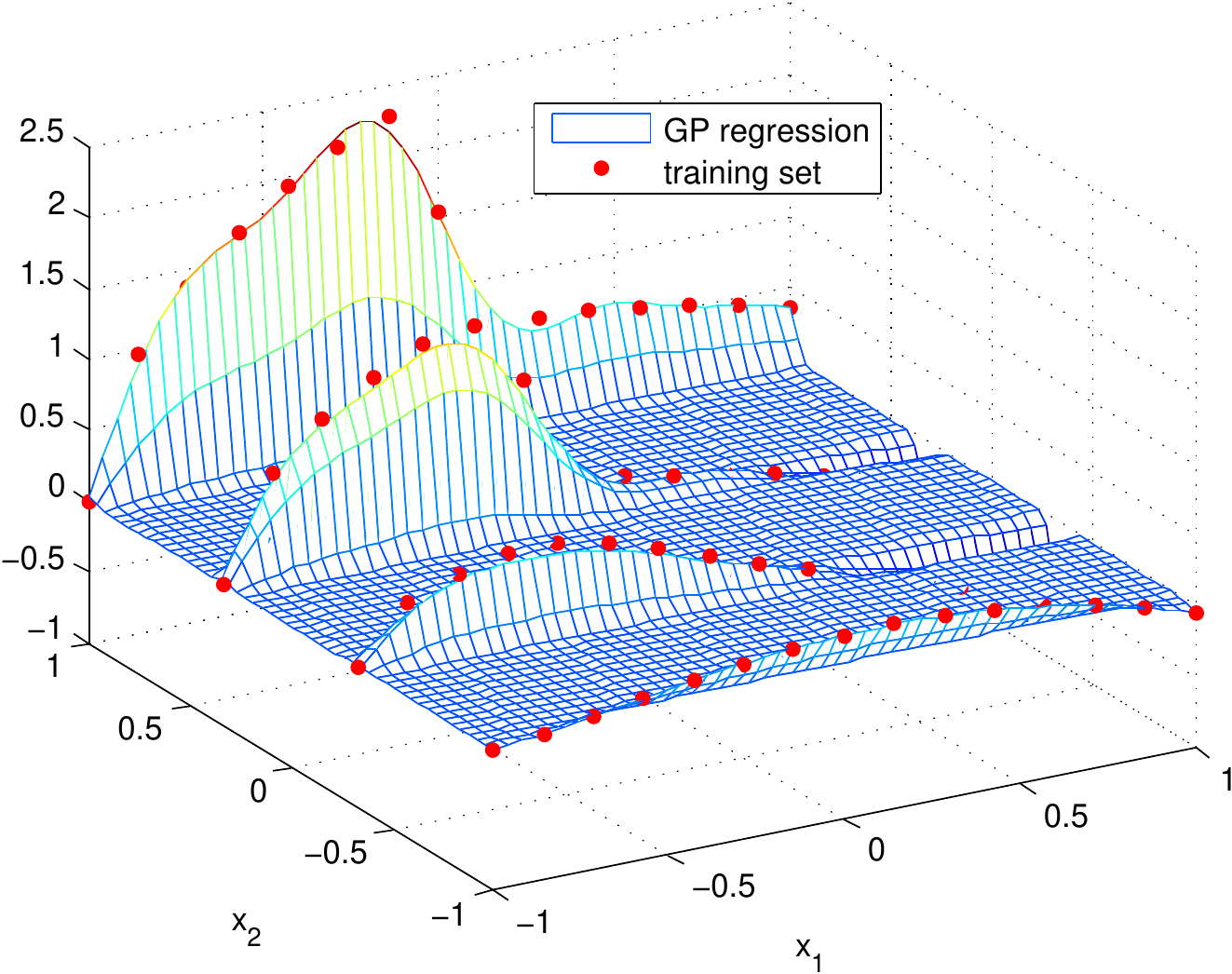}
  \caption{Degeneracy of the GP model in case of anisotropic data set.
    The length-scales for this model are $l_1 = 0.286, l_2 = 0.033$
    whereas factor sizes are $n_1 = 15, n_2 = 4$.}
  \label{fig:anisotropy_degeneracy}
\end{figure}
\begin{figure}
  \centering
  \includegraphics[width=0.4\textwidth]{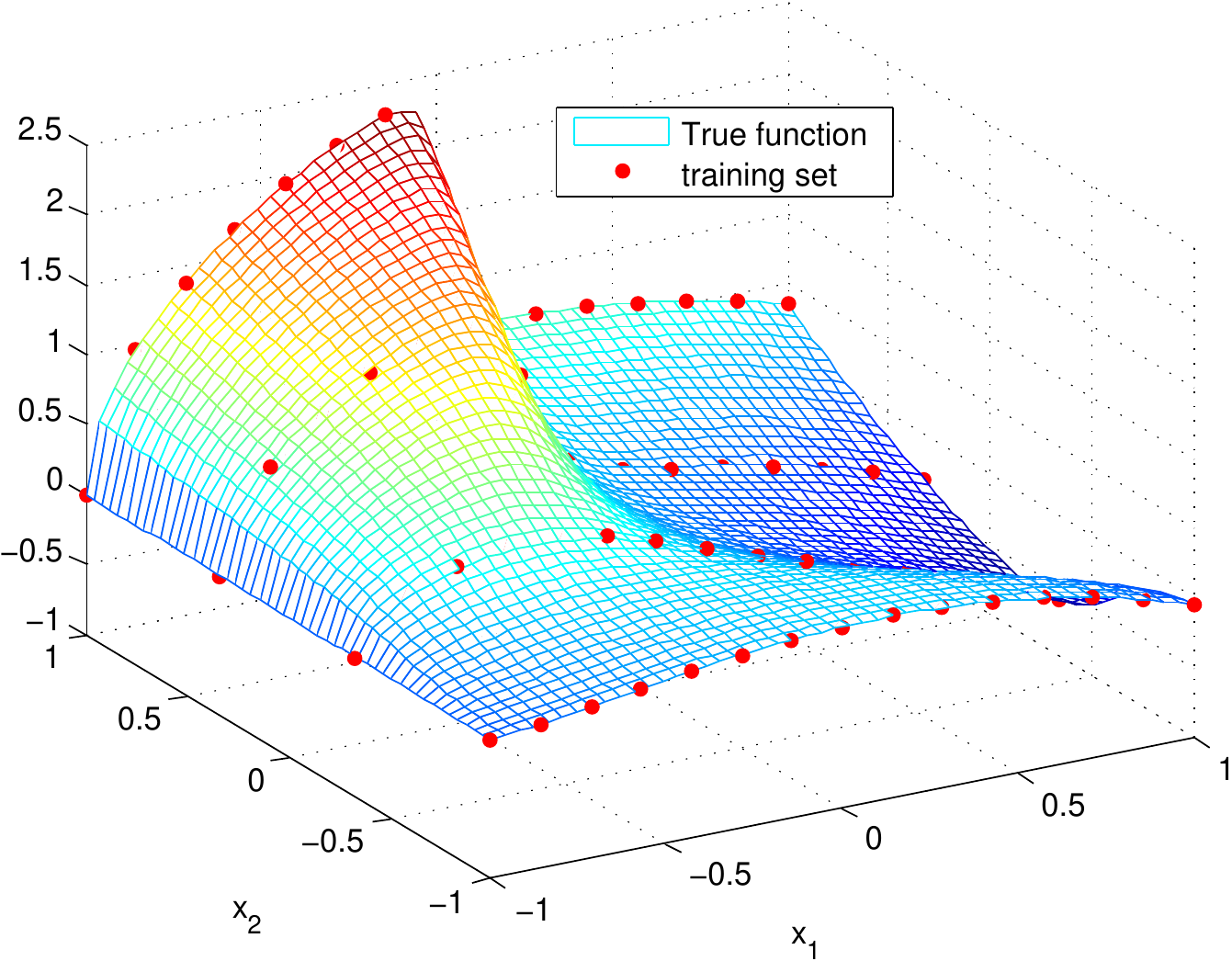}
  \caption{True function.}
  \label{fig:anisotropy_degeneracy_true}
\end{figure}
\begin{figure}
  \centering
  \includegraphics[width=0.4\textwidth]{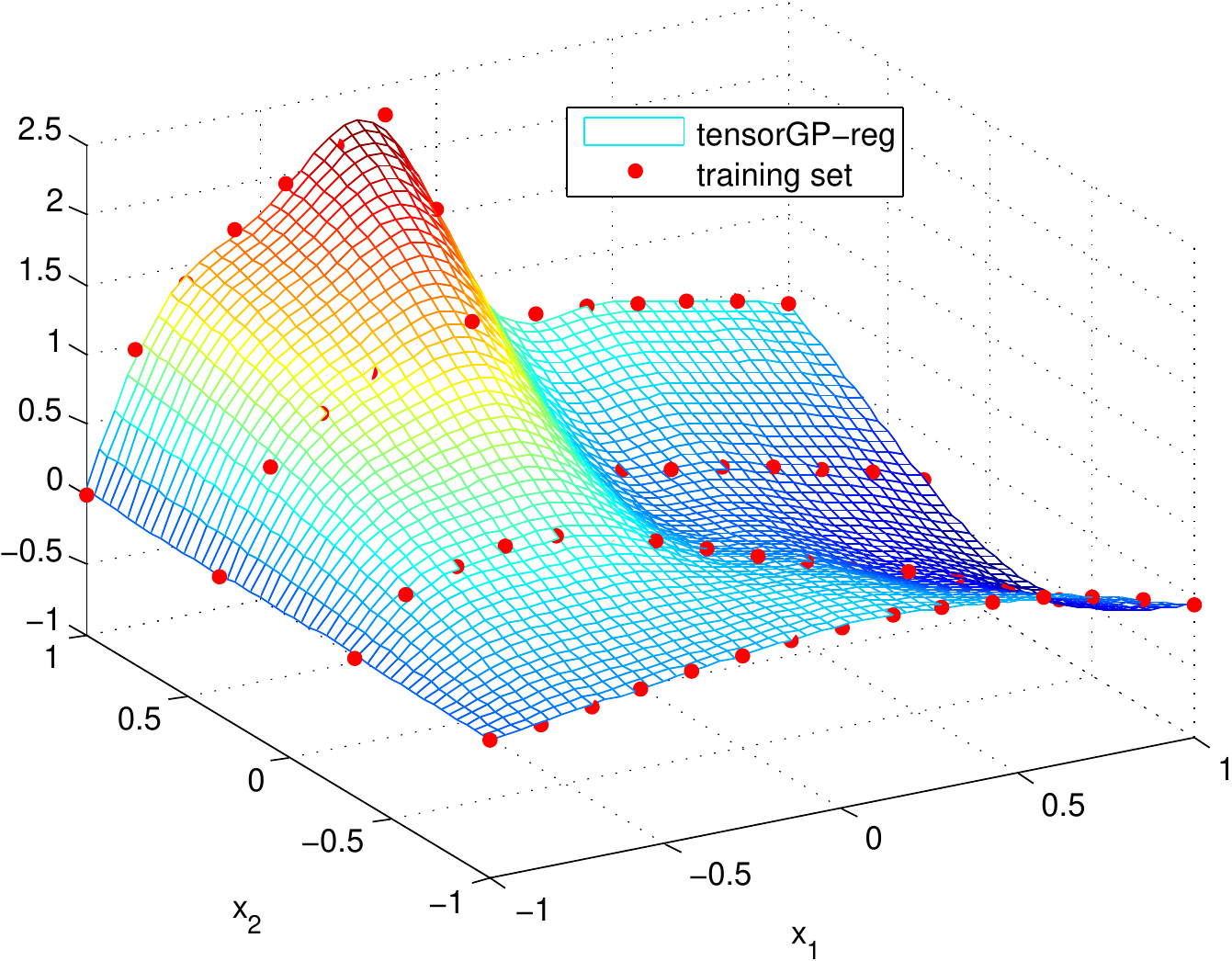}
  \caption{The GP regression with proposed prior distribution and initialization.}
  \label{fig:nondegenerate_tensorGP}
\end{figure}

\section{Experimental results}
Proposed algorithm was tested on a set of test functions \cite{lappeenranta, ETH}.
The functions have different input dimensions from 2 to 6 and the sample sizes $N$
varied from $100$ to about $200000$.
For each function several factorial anisotropic training sets were generated.
We will test the following algorithms: GP with tensor computations (tensorGP),
GP with tensor computations and prior distribution (tensorGP-reg),
the sparse pseudo-point input GP (FITC) \cite{snelson06sparsegaussian} and
Multivariate Adaptive Regression Splines (MARS) \cite{friedman1991multivariate}.
For FITC method we used GPML Matlab code \cite{gpmltoolbox}.
Number of inducing points of FITC algorithm varied from $M = 500$ for small samples (up to $5000$ points)
to $M = 70$ for large samples (about $10^5$ points) in order to obtain approximation in reasonable time
(complexity of FITC algorithm is $\mathcal{O}(M^2N)$).
For tensorGP and tensorGP-reg we adopted GPML code to use tensor operations, proposed prior distribution
and initialization.

To assess quality of approximation a mean squared error was used
\begin{equation}
  \label{eq:mse}
  {\rm MSE} = \frac{1}{N_{test}} \sum_{i = 1}^{N_{test}} (\hat{f}(\mathbf{x}_i) - f(\mathbf{x}_i))^2,
\end{equation}
where $N_{test} = 50000$ is a size of test set.
The test sets were generated randomly.

A large set of problems (number of problems is about 40) was used to test the algorithms.
Usual ``test problem vs. MSE'' plot will be confusing for large set of problems as it will look like
noisy graph.
To see a picture of the overall performance of algorithms on such test problems set
we use Dolan-Mor\'{e} curves \cite{dolanMore}.
The idea of Dolan-Mor\'{e} curves is as follows.
Let $t_{p, a}$ be an error of an $a$-th algorithm on a $p$-th problem and $r_{p, a}$
be a performance ratio
\[
r_{p, a} = \frac{t_{p, a}} {\min\limits_s(t_{p, s})}.
\]
Then Dolan-Mor\'{e} curve is a graph of $\rho_a(\tau)$ function where
\[
\rho_a(\tau) = \frac{1}{n_p}{\rm size} \{p: r_{p, a} \le \tau \},
\]
which can be thought of as a probability for the  $a$-th algorithm to have performance
ratio within factor $\tau \in \mathbb{R}_+$.
The higher the curve $\rho_a(\tau)$ is located the better works the corresponding algorithm.
$\rho_a(1)$ is a number of problems on which the $a$-th algorithm showed the best performance.

As expected tensorGP performs better than FITC as it uses all the information
containing in the training sample.
Introduced prior distribution is more suited for anisotropic data and this allows
tensorGP-reg algorithm to outperform other tested algorithms (see Figure \ref{fig:dolan_more}).

To compare run-time performances of the algorithms we plotted Dolan-Mor{'e} curves
where instead of approximation error the training time was used (see Figure \ref{fig:dolan_more_time}).

Time performance of tensorGP and tensorGP-reg is comparable to that of MARS algorithm
and outperform FITC technique.

\begin{figure}
  \centering
  \includegraphics[width=0.5\textwidth]{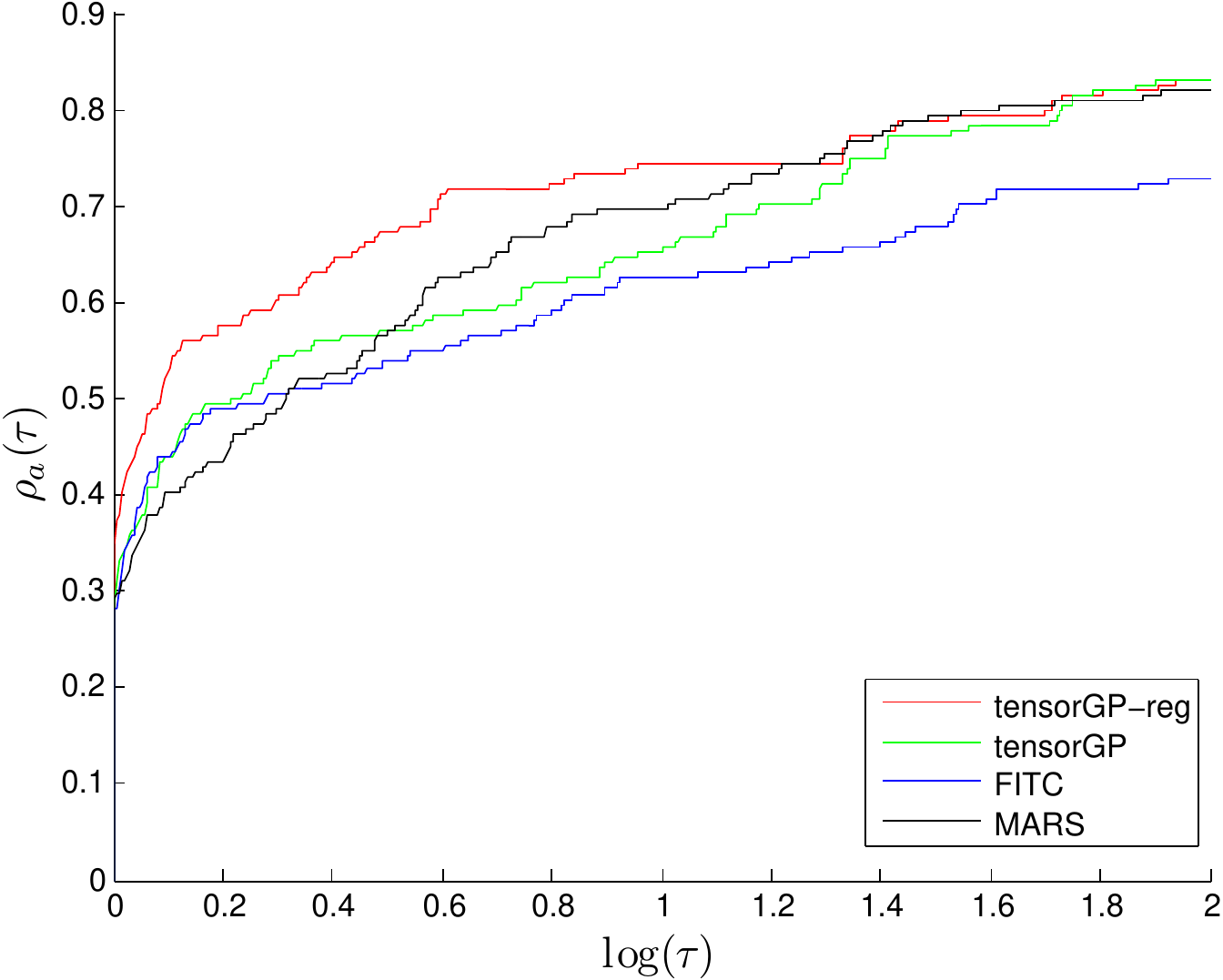}
  \caption{Approximations accuracies comparison.
    Dolan-Mor\'e curves for tensorGP, tensorGP-reg, MARS and FITC algorithms
    in logarithmic scale.
    The higher lies the curve the better performs the corresponding algorithm.}
  \label{fig:dolan_more}
\end{figure}

\begin{figure}
  \centering
  \includegraphics[width=0.5\textwidth]{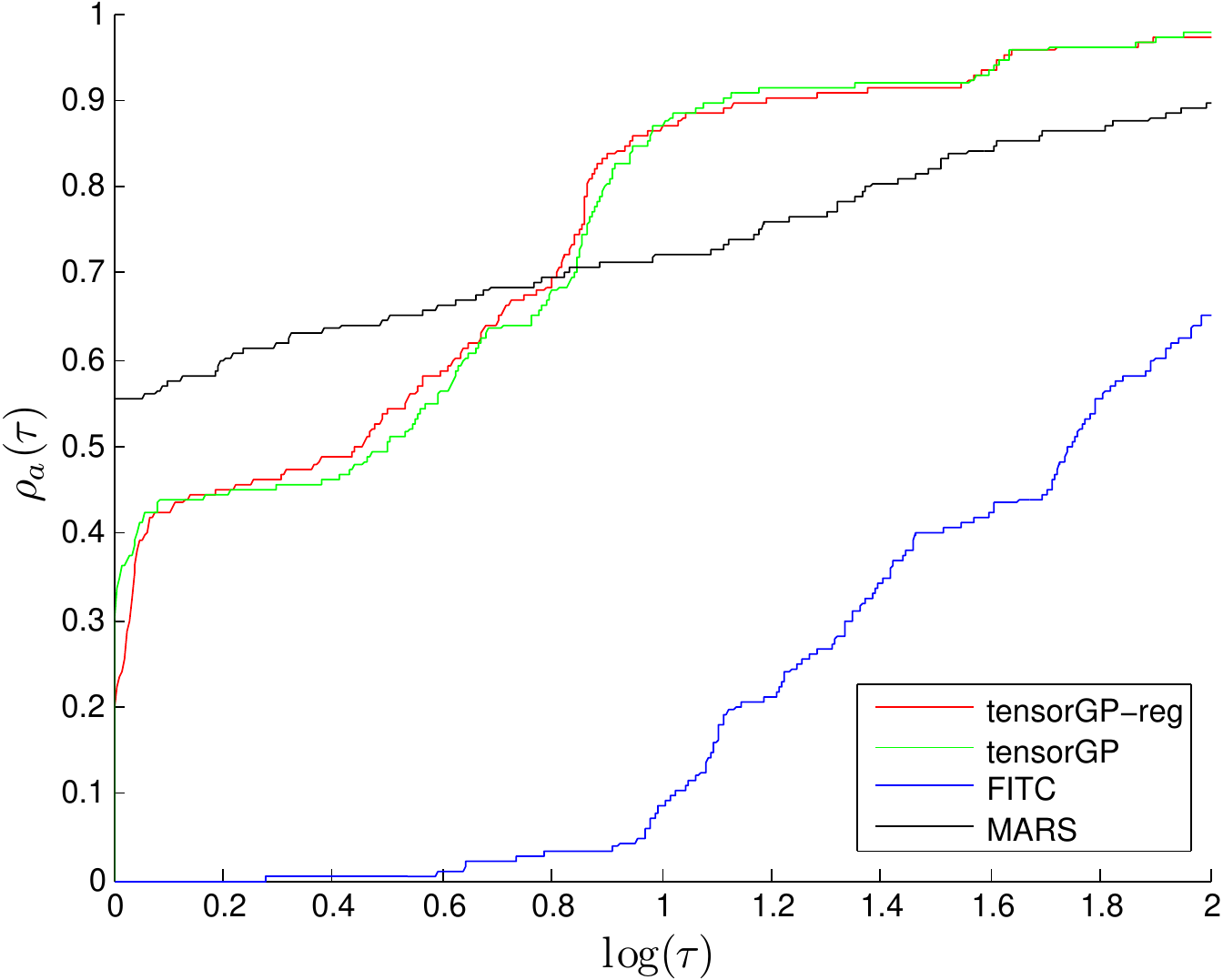}
  \caption{Run-times comparison.
    Dolan-Mor\'e curves for tensorGP, tensorGP-reg, MARS and FITC algorithms
    in logarithmic scale.
    The higher lies the curve the better performs the corresponding algorithm.}
  \label{fig:dolan_more_time}
\end{figure}

\subsection{Rotating disc problem}

In this section we will consider a real world problem of rotating disc shape design.
Such kind of problems often arises during aircraft engine design and in turbomachinery \cite{armand1995structural}.

In this problem a disc of an impeller is considered.
It is rotated around the shaft.
The geometrical shape of the disc considered here is parameterized by 6 variables
$\mathbf{x} = (h_1, h_2, h_3, h_4, r_2, r_3)$ ($r_1$ and $r_4$ are fixed), see Figures \ref{fig:rotating_disc_parametrization} and \ref{fig:rotating_disc_objectives}.
The task of an engineer is to find such geometrical shape of the disc that minimizes disc's weight and contact pressure $p_1$ between
the disc and the shaft while constraining the maximum radial stress $Sr_{max}$ to be less than some threshold.
The physical model of a rotating disc is described in \cite{armand1995structural} and it was adopted to the disc shape
presented in Figures \ref{fig:rotating_disc_parametrization}, \ref{fig:rotating_disc_objectives} in order to calculate the contact pressure $p_1$
and the maximum radial stress $Sr_{max}$.

\begin{figure}
  \center
  \includegraphics[width=0.4\textwidth]{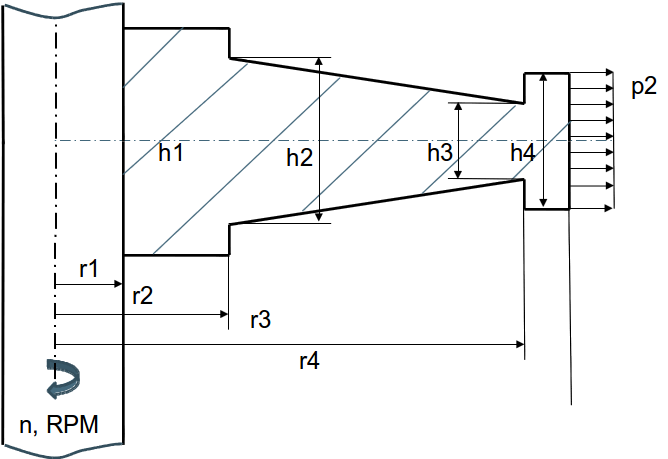}
  \caption{Rotating disc parametrization.}
  \label{fig:rotating_disc_parametrization}
\end{figure}

\begin{figure}
  \center
  \includegraphics[width=0.4\textwidth]{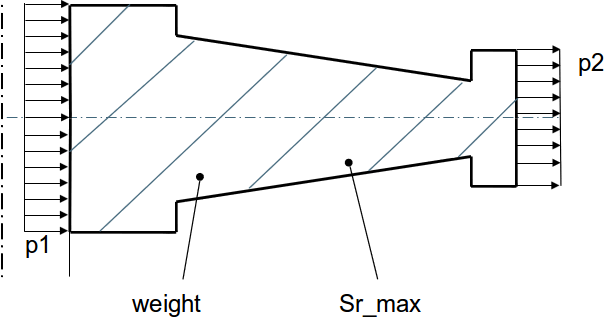}
  \caption{Rotating disc objectives.}
  \label{fig:rotating_disc_objectives}
\end{figure}

It is a common practice to build approximations of objective functions in order to analyze them
and perform optimization \cite{forrester2008surrogateModelling}.
So, we applied developed in this work tensorGP-reg algorithm and FITC to this problem.
The design of experiments was full factorial, number of points in each dimension was $[1, 8, 8, 3, 15, 5]$, i.e
$x_1$ was fixed.
Number of points in factors differ significantly and the generated data set is anisotropic.
The overall number of points in the training sample was 14 400.

Figures \ref{fig:rotating_disc_originalGP} and \ref{fig:rotating_disc_tensorGP} depict 2D slices of contact pressure
approximations along $x_5, x_6$ variables.
As you can see FITC model degenerates while tensorGP-reg provides smooth and accurate approximation.

\begin{figure}
  \includegraphics[width=0.5\textwidth]{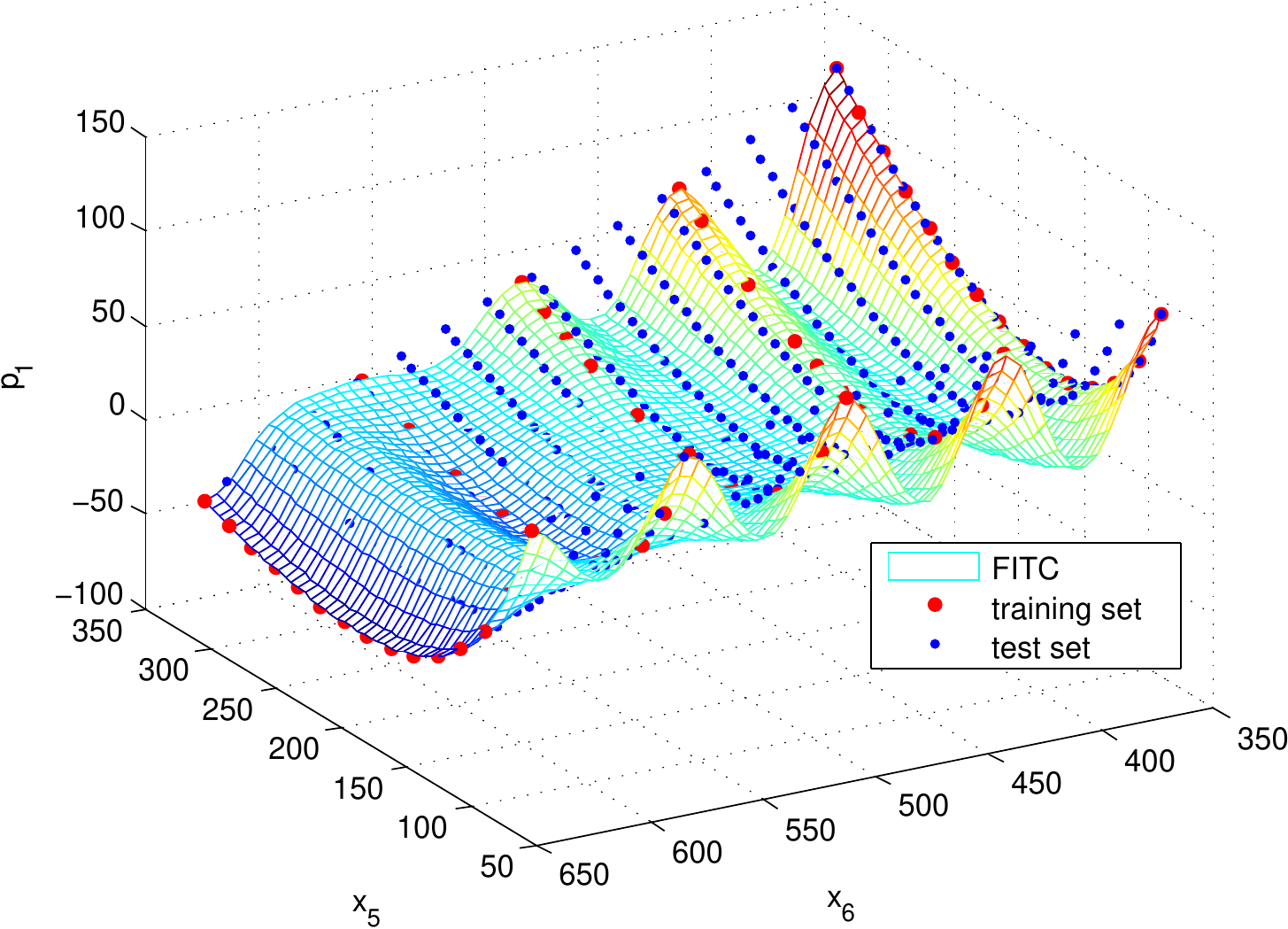}
  \caption{2D slice along $x_5$ and $x_6$ variables (other variables are fixed) of FITC approximation with 500 inducing inputs.
    It can be seen that the approximation degenerates.
  }
  \label{fig:rotating_disc_originalGP}
\end{figure}

\begin{figure}
  \includegraphics[width=0.5\textwidth]{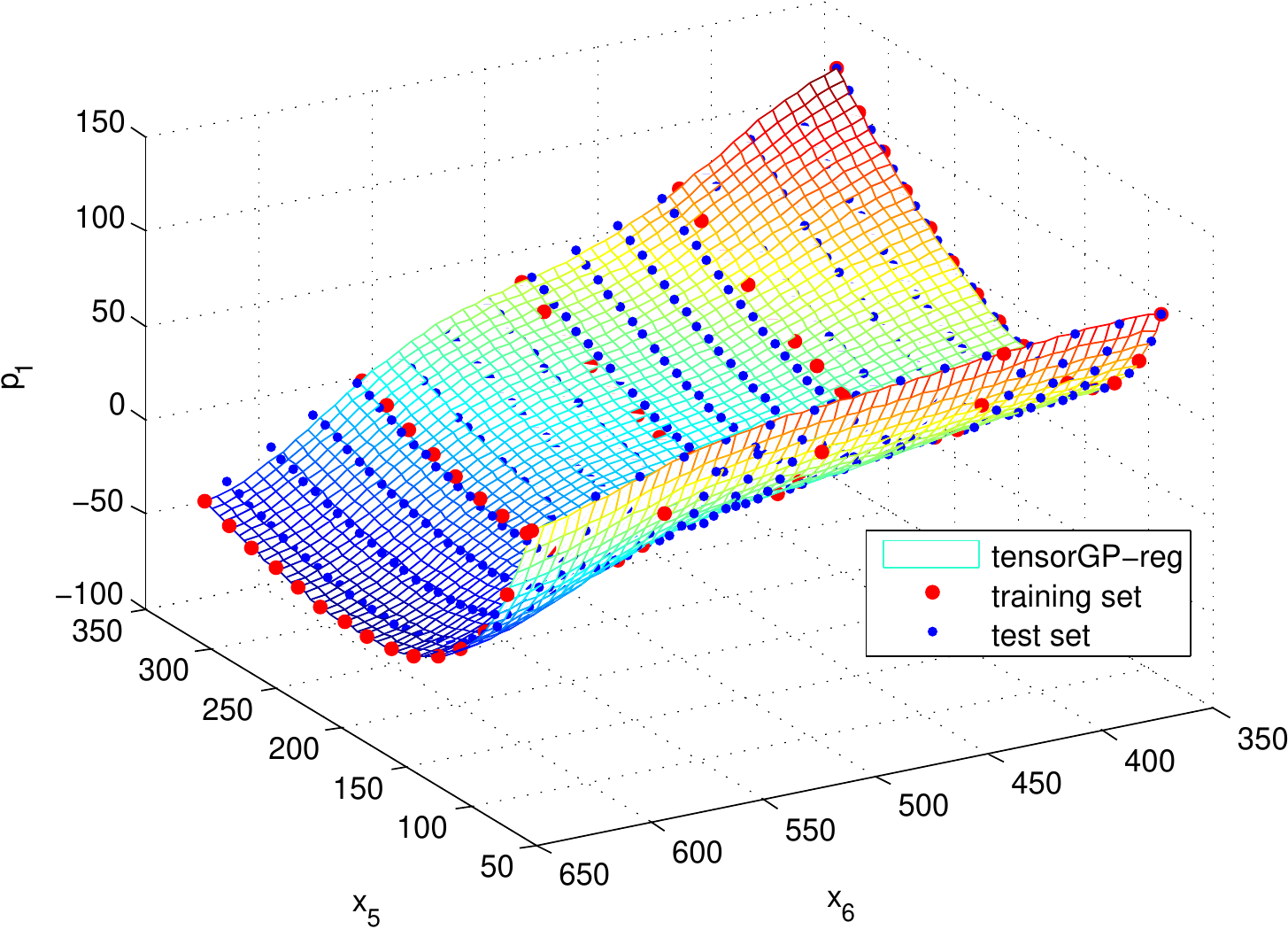}
  \caption{2D slice along $x_5$ and $x_6$ variables (other variables are fixed) of tensorGP-reg approximation.
  It can be seen that tensorGP-reg provides accurate approximation.}
  \label{fig:rotating_disc_tensorGP}
\end{figure}

\section{Conclusion}
Gaussian Processes are often used for building approximations
for small data sets.
However, knowledge about the structure of the given data set can
contain important information which allows us to efficiently
compute exact inference even for large data sets.

Introduced prior distribution combined with reasonable initialization has proven to be an efficient way to struggle
degeneracy in case of anisotropic data.

Algorithm proposed in this paper takes into account the special factorial structure of the data set and
is able to handle large multidimensional samples preserving power and flexibility
of GP regression.
Our approach has been successfully applied to toy and real problems including the rotating disc shape design.

\bibliography{tensorGP}
\bibliographystyle{icml2014}

\end{document}